\documentclass{llncs}

\usepackage{geometry}
\usepackage{amsmath,amssymb}
\usepackage{graphicx} 
\usepackage{verbatim}
\usepackage{xcolor}
\usepackage{subcaption}
\usepackage{booktabs}
\usepackage{tabularx}
\usepackage{float}
\usepackage{multirow}

\graphicspath{{figures/}} % Graphics will be here

\begin{document}

\title{ Generation of random chordal graphs using subtrees of a tree\thanks{A preliminary version of this paper has appeared in the Proceedings of the 10th International Conference on Algorithms and Complexity, CIAC 2017 \cite{CIAC2017}.} \thanks{This work is supported by the Research Council of Norway and Bo\u{g}azi\c{c}i University Research Fund (grant 11765); and T. Ekim is supported by Turkish Academy of Sciences GEBIP award.}}

\author{
	Oylum \c{S}eker\inst{1}
	\and
	Pinar Heggernes\inst{2}
	\and
	T{\i}naz Ekim\inst{1}
	\and Z.~Caner Ta\c{s}k{\i}n\inst{1}
}

\institute{
	Department of Industrial Engineering, Bo\u{g}azi\c{c}i University, Istanbul, Turkey.
	\texttt{\{oylum.seker,tinaz.ekim,caner.taskin\}@boun.edu.tr}
	\and
	Department of Informatics, University of Bergen, Norway. 
	\texttt{pinar.heggernes@uib.no}
}

\maketitle

\begin{abstract}
Chordal graphs form one of the most studied graph classes. Several graph problems that are NP-hard in general become solvable in polynomial time on chordal graphs, whereas many others remain NP-hard. For a large group of problems among the latter, approximation algorithms, parameterized algorithms, and algorithms with moderately exponential or sub-exponential running time have been designed. Chordal graphs have also gained increasing interest during the recent years in the area of enumeration algorithms. Being able to test these algorithms on instances of chordal graphs is crucial for understanding the concepts of tractability of hard problems on graph classes. Unfortunately, only few published papers give algorithms for generating chordal graphs. Even in these papers, only very few methods aim for generating a large variety of chordal graphs. Surprisingly, none of these methods is directly based on the ``intersection of subtrees of a tree'' characterization of chordal graphs. In this paper, we give an algorithm for generating chordal graphs, based on the characterization that a graph is chordal if and only if it is the intersection graph of subtrees of a tree. Upon generating a random host tree, we give and test various methods that generate subtrees of the host tree. We compare our methods to one another and to existing ones for generating chordal graphs. Our experiments show that one of our methods is able to generate the largest variety of chordal graphs in terms of maximal clique sizes. Moreover, two of our subtree generation methods result in an overall complexity of our generation algorithm that is the best possible time complexity for a method generating the entire node set of subtrees in a ``intersection of subtrees of a tree'' representation. The instances corresponding to the results presented in this paper, and also a set of relatively small-sized instances are made available online.  

\end{abstract}


\section{Introduction}\label{chapter:intro}

Algorithms particularly tailored to exploit properties of various graph classes have formed an increasingly important area of graph algorithms during the last five decades. With the introduction of relatively new 
theories for coping with NP-hard problems, like parameterized algorithms, algorithmic research on graph classes has become even more popular recently, and the number of results in this area appearing at international conferences and journals is now higher than ever.  One of the most studied graph classes in this context is the class of chordal graphs, i.e., graphs that contain no induced cycle of length 4 or more. Chordal graphs arise in practical applications from a wide variety of unrelated fields, like sparse matrix computations, database management, perfect phylogeny, VLSI, computer vision, knowledge based systems, and Bayesian networks \cite{BLS99,Gol04,pearl2014probabilistic,rose1972graph,Spinrad03}. This graph class that first appeared in the literature as early as 1958 \cite{Hajnal}, has steadily increased its popularity, and there are now more than 20 thousand references on chordal graphs according to Google Scholar. 

With a large number of existing algorithms specially tailored for chordal graphs, it is interesting to note that not much has been done to test these algorithms in practice. Very few such tests are available as published articles \cite{andreou2005generating,markenzon2008two,pemmaraju2005approximating}. In particular, there seems to be no efficient chordal graph generator available that is capable of producing every chordal graph. Most of the work in this direction involves generating chordal graphs tailored to test a particular algorithm or result \cite{andreou2005generating,pemmaraju2005approximating}. This is a clear shortcoming for the field, and it was even mentioned as an important open task at a Dagstuhl Seminar \cite{Dagstuhl}.  Until some years ago, most of the algorithms tailored for chordal graphs had polynomial running time, and testing was perhaps not crucial. Now, however, many parameterized and exponential-time algorithms exist for chordal graphs, for problems that remain hard on this graph 
class, see e.g., \cite{Sofsem2014,Pinar,DMarx,Neel}. The proven running times of such algorithms might often be too high compared to the practical running time. Just to give some examples from the field of enumeration, there are now several algorithms and upper bounds on the maximum number of various objects in chordal graphs \cite{Pinar1,Pinar2,Pinar}. However, the lower bound examples at hand usually do not match these upper bounds. Tests on random chordal graphs is a good way of getting better insight about whether the known upper bounds are too high or tight. 

In this paper, we present an algorithm for generating random chordal graphs. Our algorithm is based on the characterization that a graph is chordal if and only if it is the intersection graph of subtrees of a tree. Surprisingly, this characterization does not seem to have been directly used for random chordal graph generation earlier. Starting from a random host tree, we propose three different methods for generating subtrees of the host tree to give different neighborhood and density properties. Our algorithm, with two of these methods, can be implemented in such a way that the overall running time is best possible for an algorithm producing the entire node set of subtrees in a ``intersection of subtrees '' representation of a chordal graph. One of these fast subtree generation methods, which we call GrowingSubtree, is also the method that turns out to generate the largest variety of chordal graphs. We measure the variety using the characteristics of the maximal cliques of the generated graph, as it has been done in previous work \cite{pemmaraju2005approximating}. After proving the correctness,% and the running time,
we give extensive computational tests to demonstrate the kind of chordal graphs that our algorithm generates using each of the different subtree generation methods. We compare our methods with one another and with existing test results; we also implement one of the earlier proposed methods and include this in our tests. Note that {\sc Graph Isomorphism} is as hard on chordal graphs as on general graphs \cite{lueker1979linear}, which adds to the difficulty of producing chordal graphs uniformly at random. Our algorithm is able to generate every chordal graph with positive probability.

\section{Background, terminology and existing algorithms}
\label{chapter:literature}
In this section we give the necessary background on chordal graphs, as well as a short review of the existing algorithms for chordal graph generation. We work with simple and undirected graphs, and we use standard graph terminology. 
In particular, for a given graph $G$, we denote its vertex set by $V(G)$, and edge set by $E(G)$. We let $n=|V(G)|$ and $m=|E(G)|$. The set of {\it neighbors}, or the {\it neighborhood}, of a vertex is the set of vertices adjacent to it. The size of the neighborhood of a vertex $x$ is the {\it degree} of $x$, denoted by $d(x)$. The neighborhood of a set $X$ of vertices is the union of the neighborhoods of the vertices in $X$, excluding $X$ itself. 

Let $F=\{S_1, S_2, \ldots, S_n\}$ be a family of sets from the same universe. A graph $G$ is called an {\it intersection graph of} $F$ if there is a bijection between the set of vertices $\{v_1, v_2, \ldots, v_n\}$ of $G$ and the sets in $F$ such that $v_i$ and $v_j$ are adjacent if and only if $S_i \cap S_j \neq \emptyset$, for $1 \leq i,j \leq n$. In the special case where there is a tree $T$ such that each set in $F$ corresponds to the vertex set of a subtree of $T$,  then $G$ is called the {\it intersection graph of subtrees of a tree}.

A {\it clique} is a set of vertices that are pairwise adjacent.
An ordering $(v_1, v_2, \ldots, v_n)$ of the vertices of a graph is a {\it perfect elimination order (peo)} if the set of higher numbered neighbors of each vertex forms a clique.  
A {\it maximal clique} is a clique $C$ such that no set of vertices properly containing $C$ is a clique.  Let  $K$ be the set of maximal cliques of a graph $G$. A tree $T$ with a bijection between its vertex set and the cliques in $K$ is called a {\it clique tree} of $G$ if for every vertex $v$ of $G$, the set of vertices of $T$ that correspond to the cliques containing $v$ induce a connected subtree of $T$.

A graph is {\it chordal} if it contains no induced cycle of length 4 or more. A chordal graph on $n$ vertices has at most $n$ maximal cliques \cite{Dirac}. Chordal graphs have many different characterizations. For our purposes, the following will be sufficient. 

\begin{theorem}
[\cite{buneman1974characterisation,fulkerson1965incidence,gavril1972algorithms,gavril1974intersection}]
\label{chordal-big}
Let $G$ be a graph. The following are equivalent.
\begin{itemize}
 \item $G$ is chordal.
 \item $G$ has a perfect elimination order. 
 \item $G$ is the intersection graph of subtrees of a tree.
 \item $G$ has a clique tree. 
\end{itemize}
\end{theorem}

Especially the last two points of Theorem \ref{chordal-big} are crucial for our algorithm and its implementation. To make sure that there is no confusion between the vertices of $G$ and the vertices of a tree or a clique tree, we will from now on refer to vertices of a tree as {\it nodes}.

Rose, Tarjan, and Lueker \cite{rose1976algorithmic} gave an algorithm called Maximal Cardinality Search (MSC) that creates a perfect elimination order of a chordal graph in time $O(n+m)$. Blair and Peyton \cite{BP} gave a modification of MCS to list all the maximal cliques of a chordal graph in time $O(n+m)$. Implicit in their proofs is the following well-known fact, which can be characterized as folklore. 

\begin{lemma}[\cite{BP,rose1976algorithmic}]
\label{chordal-cliquesum}
The sum of the sizes of the maximal cliques of a chordal graph is $O(n+m)$.
\end{lemma}

Next, we briefly mention the algorithms for generating chordal graphs from the works of Andreou, Papadopoulou, Spirakis, Theodorides, and Xeros \cite{andreou2005generating}; Pemmaraju, Penumatcha, and Raman \cite{pemmaraju2005approximating}; and Markenzon, Vernet, and Araujo \cite{markenzon2008two}. Some of these algorithms create very limited chordal graphs, which is either mentioned by the authors or clear from the algorithm. Thus, in the following we only mention the algorithms that are general enough to be interesting in our context. 

It should also be noted that the purpose of Andreou et al.~\cite{andreou2005generating} is not to obtain general chordal graphs, but rather chordal graphs with a known bound on some parameter. One of the algorithms that they propose starts from an arbitrary graph and adds edges to obtain a chordal graph. How the edges are added is not given in detail, however it should be noted that there are many algorithms for generating a chordal graph from a given graph by adding a minimal set of edges and their running time is usually $O(nm)$ \cite{heggernes2006minimal}. Andreou et al.~\cite{andreou2005generating} do not report on the quality of chordal graphs obtained by this method.

We highlight below the algorithms that are the most promising with respect to generating random chordal graphs. In addition to these, there is an $O(n^2)$-time  algorithm by Markenzon et al.~\cite{markenzon2008two} that generates a random tree and adds edges to this tree until a chordal graph with desired edge density is obtained. However, no test results about the quality of the generated graphs is given. 

\medskip

\noindent
{\bf Alg 1 \cite{andreou2005generating}.} The algorithm constructs a chordal graph by using a peo. At every iteration, a new vertex is added and made adjacent to a random selection of already existing vertices. Then necessary edges are added to turn the neighborhood of the new vertex into a clique such that a given maximum degree is not exceeded. No test results are given in the paper about the quality of the chordal graphs this algorithm produces. As we found the algorithm interesting, we have implemented it, and we compare the resulting graphs to those generated by our algorithm in Section \ref{chapter:experiments}. 
\medskip

\noindent
{\bf Alg 2 \cite{markenzon2008two,pemmaraju2005approximating}.} The algorithm starts from a single vertex. At each subsequent step, a clique $C$ in the existing graph is chosen at random, and a new vertex is added adjacent to exactly the vertices of $C$. The inverse of the order in which the vertices are added is a peo of the final graph. It is observed by the authors of both papers that this procedure results in chordal graphs with approximately $2n$ edges experimentally. They propose the following changes:

{\bf Alg 2a \cite{markenzon2008two}} modifies the above generated graph by randomly choosing maximal cliques that are adjacent according to the clique tree and merging these until desired edge density is obtained. Some test results about graphs generated by Alg 2a are provided in \cite{markenzon2008two}. Although these tests are not as comprehensive as the ones we give on our algorithms in Section \ref{chapter:experiments}, we compare our results to those of \cite{markenzon2008two} as best we can. The running time of Alg 2a is $O(m+n\alpha(2n,n))$. 
 
{\bf Alg 2b \cite{pemmaraju2005approximating}} is a modification of Alg 2 in a different way: instead of randomly choosing a clique, a maximum clique is chosen and a random subset of it is made adjacent to the new vertex. Although test results for Alg 2b are provided in \cite{pemmaraju2005approximating}, the authors acknowledge that the produced graphs are still very particular with very few large maximal cliques and many very small maximal cliques. For this reason, we do not include Alg 2b in our comparisons.

\section{Generating chordal graphs using subtrees of a tree}\label{chapter:algo}

We find it surprising that the intersection graph of subtrees of a tree characterization of chordal graphs has not been used directly for their generation. Of course, since all characterizations of a chordal graph are equivalent, even the existing algorithms mentioned above could be interpreted as based on any of these characterization. Especially the algorithms based on clique trees can easily be translated to generate subtrees of a tree. However, none of these algorithms generate random subtrees of a randomly generated tree to produce the resulting chordal graph. One reason could be that this does not give a direct way to decide the number of edges in the generated graph. We will see that edge density can be regulated by adjusting the sizes of the generated subtrees. Let us first observe the following.

\begin{lemma}
\label{treetree}
For every chordal graph $G$ on $n$ vertices, there is a tree $T$ on $n$ nodes, and a set of $n$ subtrees of $T$, such that $G$ is the intersection graph of these subtrees.
\end{lemma}

\begin{proof}
 Let $G$ be a chordal graph on $n$ vertices and let $T'$ be a clique tree of $G$. Let us call the vertices of $G$: $v_1, v_2, \ldots, v_n$. Define subtree $T'_i$ to be the subtree of $T'$ that corresponds to the nodes (maximal cliques) that contain vertex $v_i$, for $1 \le i \le n$. By the definition of a clique tree, $T'$ has at most $n$ nodes and each $T'_i$ is a connected subgraph of $T'$. If $T'$ has fewer than $n$ nodes, we can add new nodes adjacent to arbitrary nodes of $T'$ until we get a new tree $T$ with $n$ nodes. The subtrees stay the same. As two vertices are adjacent in $G$ if and only if they appear together in a clique, $G$ is the intersection graph of subtrees $T'_1, \ldots, T'_n$ of $T$. \qed
\end{proof}

Based on Lemma \ref{treetree}, we are ready to present our main algorithm for generating chordal graphs on $n$ vertices: 

\bigskip

{\bf Algorithm ChordalGen} 

\smallskip
{\bf Input:}  $n$ and one or two other numbers to guide the number of resulting edges.

{\bf Output:} A chordal graph $G$ on $n$ vertices and $m$ edges. 

\smallskip
1. Generate a tree $T$ on $n$ nodes uniformly at random. 

2. Generate $n$ non-empty subtrees of $T: \{T_1, \ldots, T_n\}$. 

3. Output as $G$ the intersection graph of $\{V(T_1), \ldots, V(T_n)\}$. 

\bigskip

By Theorem \ref{chordal-big}, the graph generated by Algorithm ChordalGen is chordal. By Lemma \ref{treetree}, our algorithm can create any chordal graph. Later in this section we will describe three different methods for generating the $n$ subtrees in Step 2. Each method will take one or two parameters to guide the average size of the generated subtrees, with the purpose of controlling the resulting number of edges in $G$. Our algorithm is flexible in the sense that additional ways to generate the subtrees can be suggested and tested later. 

\medskip

We need to evaluate the order of $\sum_{i=1 }^{n} |V(T_i)|$ to analyze the running time of Algorithm ChordalGen. Let us explain at this point that in the preliminary version of this paper \cite{CIAC2017}, $\sum_{i=1 }^{n} |V(T_i)|$ was mistakenly claimed to be linear in the size of the generated chordal graph, that is $O(n+m)$. However, it turns out that $\sum_{i=1 }^{n} |V(T_i)|$ is $\Omega(mn^{1/4})$ as shown in what follows.

\begin{lemma}\label{lem:key}
Let $T$ be a tree on $n$ nodes, and let $T_1,\ldots ,T_n$ be $n$ subtrees of $T$ whose intersection graph is a chordal graph on $n$ vertices and $m$ edges. Then  $\sum_{i=1 }^{n} |V(T_i)|$is $\Omega(mn^{1/4})$.
\end{lemma}

\begin{proof}
Let $G$ be a chordal graph that is the intersection graph of  $\{V(T_1),\ldots ,V(T_n)\}$, such that every subtree $T_i$ is associated with a vertex $v_i$ of $G$.   
Choose an integer $p\geq  2$. Let $n = p^4, n' = p^2 + 1, s = p^2 - p$, and $s_{n'} = p^3$. Note that
$n > n'$. Let $T$ be a path on $n$ nodes $x_1, \ldots x_n$. Let $T_i$ be the subtree of $T$ consisting of the single node $x_i$, for $i= 1, \ldots , n'-1$. Let $T_{n'}$ be the subpath of $T$ consisting of the nodes  from $n'$ to $n$. The representation contains $s$ copies of $T_i$ for every $i=1, \ldots ,n'-1$ and also $s_{n'}$ copies of $T_{n'}$. Therefore, the number of subtrees is
$$ s (n' - 1) + s_{n'} = (p^2 - p) p^2 + p^3 = p^4 = n.$$
The corresponding graph is the disjoint union of a $K_{s_{n'}}$ and $n' - 1$ copies of $K_s$. The number $m$ of edges of this graph is
$$ m = (n' - 1)\frac{s(s - 1)}{2} + \frac{s_{n'}(s_{n'} - 1)}{2} = p^2 \frac{(p^2 - p)(p^2 - p - 1)}{2}
+ \frac{p^3(p^3 - 1)}{2} = p^6 - p^5.$$
On the other hand, the total size of the subtrees is
$$s (n' - 1) + s_{n'}(n - n' + 1) = (p^2 - p)p^2 + p^3(p^4 - p^2) = p^7 - p^5 + p^4 - p^3 \geq p^7 - p^5.$$
Therefore, 
$$\frac{\sum_{i=1 }^{n} |V(T_i)|}{m} \geq \frac{p^7 - p^5}{p^6 - p^5} = p + 1 > n^{1/4}.$$
Since the inequality holds for infinitely many values of $p$, we conclude the proof. \qed
\end{proof}

We are now ready to state our main result about Algorithm ChordalGen.

\begin{theorem}
\label{runningtime}
Algorithm ChordalGen generates a chordal graph with $n$ vertices in time $O(S+\sum_{i=1 }^{n} |V(T_i)|)$, where $O(S)$ is the time required to generate $n$ subtrees of a host tree on $n$ nodes.
\end{theorem}

\begin{proof}
By Theorem \ref{chordal-big},  graph $G$ generated by Algorithm ChordalGen is chordal. Clearly $G$ has $n$ vertices. Let $m$ be the number of edges of $G$. Let us go through the steps of the algorithm to argue for the running time.

Step 1. For the generation of a random host tree $T$ on $n$ vertices, we use the following method by Rodionov and Choo \cite{Rodionov}, which can easily be implemented in $O(n)$ time: start with a tree $T$ that contains only one node. Then repeat $n-1$ times the following: pick a random vertex $x$ of $T$ and add a new vertex adjacent to it. 

Step 2. We generate $n$ subtrees $T_1,\ldots ,T_n$ of $T$ using a subtree generator. According to the premises of the theorem, this adds   $O(S)$  to the overall time.
 
Step 3. The sum of the sizes of the generated subtrees is $\sum_{i=1 }^{n} |V(T_i)|$. Let us now explain how we can obtain in time $O(n+m)$ the chordal graph which is the intersection graph of these subtrees. For every node $x$ of $T$, let us define the following list: $C_x = \{v_j \mid T_j$ contains $x \}$, i.e., vertices of $G$ whose corresponding subtrees contain node $x$ of $T$. Let ${\cal C}= \{C_x \mid x \in V(T)\}$. Observe that every set in $\cal{C}$ is a clique of $G$, and $\cal{C}$ contains all maximal cliques of $G$. However, some of the cliques in $\cal{C}$ may not be maximal.
If we blow up every node $x$ of $T$ to the set $C_x$, we get a tree which is almost a clique-tree of $G$; this procedure can clearly be done in $O(n+m)$ time. In the resulting tree, non-maximal cliques simply need to be deleted or merged into the maximal ones. By the methods described by Blair and Peyton \cite{BP} it is possible to turn $T$ into a proper clique tree for $G$ in time $O(n+m)$. Thus, in total $O(\sum_{i=1 }^{n} |V(T_i)|+n+m)$ time we both have a representation of our output graph $G$ and a list of maximal cliques of it. It could, however, be desirable to output an adjacency list representation for $G$. Markenzon et al.~\cite{markenzon2008two}, using the methods of Blair and Peyton \cite{BP}, explain how this can be done in $O(n+m)$ time. By Lemma \ref{lem:key}, the overall running time of this step is $O(\sum_{i=1 }^{n} |V(T_i)|)$.  \qed
\end{proof}

Now we are ready to give the details of how the subtrees of Step 2 are generated. In the subsections below, we present three methods for generating $n$ subtrees of $T$. Then, in Section \ref{chapter:experiments}, we will test our algorithm with each of these methods and compare the results with each other, as well as with Alg 1 and Alg 2a.

\subsection{Algorithm GrowingSubtree}
\label{section:blooming_subtree}
This algorithm takes as input a tree $T$ on $n$ nodes, and an integer $k$, and generates $n$ subtrees of $T$ of average size $ \frac{k+1}{2}$. In our test results, we give both $k$ and the resulting number of edges, $m$, to give an indication of how $k$ affects the density of the generated graph. 

For each subtree $T_i$, the algorithm picks a size $k_i$ randomly from $[1, k]$.  Then a random node of $T$ is chosen as the single node of $T_i$ to start with. In each of the subsequent $k_i-1$ iterations, we pick a random node of $T$ in the neighborhood of $T_i$ and add it to $T_i$.

\medskip

{\bf Algorithm GrowingSubtree} 

\smallskip
{\bf Input:}  A tree $T$ on $n$ nodes and a positive integer $k \le n$

{\bf Output:}  A set of $n$ non-empty subtrees of $T$ of average size $\frac{k+1}{2}$

\smallskip
{\bf for} $i=1$ {\bf to} $n$ {\bf do}

\quad Select a random node $x$ of $T$ and set $T_i=\{x\}$

\quad Select a random integer $k_i$ between 1 and $k$

\quad {\bf for} $j=1$ {\bf to} $k_i-1$ {\bf do}

\quad \quad Select a random node $y$ of $T_i$ that has neighbors in $T$ outside of $T_i$

\quad \quad Select a random neighbor $z$ of $y$ outside of $T_i$ and add $z$ to $T_i$

\smallskip
Output $\{T_1, T_2, \ldots, T_n\}$ \\

\begin{lemma}
\label{growingsubtree-runtime}
The running time of Algorithm GrowingSubtree is $O(\sum_{i=1 }^{n} |V(T_i)|)$.
\end{lemma}

\begin{proof}
Observe first that each subtree $T_i$ can be represented simply a list of nodes of $T$. 
We show that after an initial $O(n)$ preprocessing time, each subtree $T_i$ can be generated in time $O(|V(T_i)|)$. For this, we need to be able to add a new node to $T_i$ in constant time, at each of the $k_i - 1$ steps. 
	
	As selecting random elements in constant time is easier when accessing the elements of an array directly by indices, we start with copying the nodes of $T$ into an array $A$ of size $n$, and copying the adjacency list of each node $x$ into an array $A_x$ of size $d(x)$. This can clearly be done in total time $O(n)$ since $T$ is a tree. 
	
	In general, selecting an unselected element of a set at random can be done easily in constant time if the set is represented with an array. Let us say we have an array $S$ of $t$ elements. We keep a separation index $s$ that separates the selected elements from the not selected ones. At the beginning $s$ is 1. At each step, we generate a random integer $r$ between $s$ and $t$. $S[r]$ is our randomly selected element. Then we swap the elements $S[s]$ and $S[r]$ and increase $s$ by 1. 
	
	We can use this method both for selecting a node $y$ of $T_i$ that still has neighbors outside and for selecting a neighbor $z$ of $y$ that has not yet been selected. For the latter, whenever we select a neighbor $z$ of $y$, we move $z$ to the first part of the array $A_x$ using swap. When the separation index reaches the degree of $y$ then we know that $y$ should not be selected to grow the subtree $T_i$ at later steps. Representing $T_i$ with an array of size $k_i$, we can use the same trick to move $y$ to a part of the array that we will not select from. Also, when $z$ is added, we can check whether it is a leaf in $T$ in constant time, and immediately move it to the irrelevant part of the array for $T_i$ if so, since $z$ can then not be used for growing  $T_i$ at later steps. It is sufficient to check that $z$ is a leaf of $T$, because otherwise it must have neighbors outside of $T_i$, since $T$ is a tree and we cannot have cycles. When the generation of $T_i$ is finished, the separation indices of each 
	of its nodes should be reset before we start generating $T_{i+1}$. The adjacency arrays need not be reorganized, as we will anyway be selecting neighbors at random.
	
	Note that we do not need this trick to select an initial node $x$ of each subtree $T_i$, because we should indeed be able to select the same node several times (and grow another subtree from it perhaps in a different way).
	
	With the described method, each step of Algorithm GrowingSubtree takes $O(1)$ time, in addition to initial $O(n)$ time to copy the information into appropriate arrays. Thus the total running time is $O(\sum_{i=1 }^{n} |V(T_i)|)$. 
\qed
\end{proof}

Lemma \ref{growingsubtree-runtime} together with Theorem \ref{runningtime} gives the following.

\begin{corollary}
Algorithm ChordalGen, using the GrowingSubtree method, runs in time $O(\sum_{i=1 }^{n} |V(T_i)|)$.
\end{corollary}

We will see in Section \ref{chapter:experiments} that this method generates chordal graphs with the most even distribution of maximal clique sizes.

\subsection{Algorithm ConnectingNodes}
\label{section:thres_prob_vertices}

This algorithm takes as input a tree $T$ on $n$ nodes, and a positive real number $\lambda$. 
The purpose of the parameter $\lambda$ is to guide the desired number of resulting edges in the graph.
To generate each subtree $T_i$, we first select $k_i$ nodes of $T$, where $k_i$ is a random integer generated by making use of $\lambda$. $T_i$ is then generated to be the minimal subtree that contains the selected $k_i$ nodes. This implies that a subtree will most likely have many more nodes than those selected initially, and this must be taken into consideration when choosing $\lambda$. In our test results, we give both $\lambda$ and the resulting number of edges, $m$, to give an indication of how $\lambda$ affects the density of the generated graph.

In the following algorithm, we will make use of the standard Breadth First Search (BFS) algorithm from an arbitrary node $r$ of $T$. We will then treat $T$ as a rooted tree with root $r$, and speak about parent-child relation in the standard way, with respect to root $r$. The BFS level of a vertex is simply the distance of that vertex from $r$.

\medskip 

{\bf Algorithm ConnectingNodes} 

\smallskip
{\bf Input:}  A tree $T$ on $n$ nodes and a positive real number $\lambda$

{\bf Output:}  A set of $n$ subtrees of $T$

\smallskip

Let $r$ be an arbitrary node of $T$

Perform BFS from $r$, and identify the parent $P(x)$ and the BFS level $l(x)$ of each node $x$

$L \gets \emptyset$

{\bf for} $i=1$ {\bf to} $n$ {\bf do}

\quad $T_i \gets \emptyset$

\quad Select a random integer $k_i$ from Poisson distribution with mean $\lambda$  

\quad {\bf if} $k_i=0$  {\bf then} $k_i \gets 1$

\quad {\bf else if} $k_i > n$  {\bf then} $k_i \gets n$

\quad Select $k_i$ random nodes from $T$ to form $T_i = \{x_1,\ldots,x_{k_i}\}$ and $L = \{x_1,\ldots,x_{k_i}\}$

\quad $d \gets \max_{x \in T_i} l(x)$  

\quad {\bf while} $|L| > 1$

\qquad {\bf for all} $x \in L$ such that $l(x) = d$ {\bf do}

\qquad \quad $T_i \gets T_i \cup P(x) $, $L \gets (L \setminus \{x\}) \cup P(x) $

\qquad $d \gets d-1 $

\smallskip
Output $\{T_1, T_2, \ldots, T_n\}$ \\

For each subtree $T_i$, we first generate a random integer $k_i$ by making use of Poisson distribution with mean $\lambda$. Poisson distribution is a discrete probability distribution widely used to model number of occurrences of an event over a specified domain such as time, space etc. In our case, the domain is the host tree, and an event is selection of a node from the host tree. The parameter of this distribution is $\lambda$ and it is the average rate of event occurrences, which implies that the initial $k_i$ values will tend to increase on the average as $\lambda$ increases. The set of possible values a Poisson random variable can take is nonnegative integers, regardless of the value of $\lambda$. However, the minimum and maximum number of nodes that a subtree of an $n$-node host tree can contain are 1 and $n$ respectively. Therefore, we equate $k_i$ to 1 if it is zero, and to $n$ if it is greater than $n$. In this method, the only reason why we use Poisson distribution is that we were not able to achieve a good precision for edge density by picking a random integer uniformly at random from a given interval. To generate the minimal subtree that contains the $k_i$ selected nodes, we make use of the nodes' parent and level information retrieved during BFS. Our key observation is that the minimal subtree must contain the parents of all selected nodes at some level $d$ if there are other selected nodes at levels less than or equal to $d$. Using this idea, we add a node to $T_i$ only when an edge incident to it has to be in $T_i$ to join a node to the others. We start from the highest level (highest distance from $r$) and proceed by moving toward the root until all nodes in the selection become connected. The set $L$ keeps the unprocessed nodes yet to be connected to form the subtree $T_i$. At each step, we consider the nodes in $L$ that are at the same level, which are going to be joined as we move through the levels. Once the parents of those nodes are identified, we are done with level $d$ and there is no need to reconsider nodes at level $d$ any further. Therefore, we replace those nodes with their parents in $L$, which are to be considered at the next step. Parent nodes are also added to $T_i$ because they lie on the paths that connect the node selection. Afterwards, we move to the next level and apply the same procedure. This process continues until a single node is left in $L$, which simply means that we have a node set that has been linked at a single node already and that $T_i$ includes all nodes of the subtree that minimally connects the randomly selected node set.

\begin{lemma}
	\label{connectingnodes-runtime}
	The running time of Algorithm ConnectingNodes is $O(\sum_{i=1}^{n} |V(T_i)|)$.
	
\end{lemma}

\begin{proof}

Rooting $T$ from an arbitrary node, and determining the parent $P(x)$ of each node $x$ in $T$ as well as its level $l(x)$ with respect to the root node, takes $O(n)$ time in total for all nodes using BFS. The set $L$ is represented by an array of lists with length equal to $\max_{x\in T}{l(x)}$. Each index $d$ of $L$ represents a list of unprocessed nodes having $l(x) = d$. The lists in $L$ will be empty initially. If a node $x$ at level $d$ is selected, we add $x$ to the list at index $d$ of $L$ in constant time.
We also need to keep an $n$-dimensional boolean array $B$, which will be comprised of zeros at first, in order to check whether a node already exists in $L$. If a node $x$ is chosen, the element at index $x$ of $B$ will be set to one in constant time. Note that $B$ is a different representation of the set of nodes in $L$. Both $L$ and $B$ are initialized only once at the start of the algorithm after performing the BFS. Since the number of levels is at most $n$, initialization of $L$ and $B$ can be done in time $O(n)$. We represent $T_i$ as a list of nodes, which can be initialized in $O(1)$. 
For each subtree, we generate a random integer $k_i$ by making use of Poisson distribution with mean $\lambda$. Generation of some random integer $x$ from Poisson distribution normally takes $O(x)$ time; starting from zero, the value of $x$ is incremented one by one until the stopping condition is met \cite{knuth1969art}. However, since we need $k_i$, which is the number nodes to be selected, to lie between 1 and $n$, we terminate the process if we reach $n$ before stopping condition is met, and we set $k_i$ to 1 if the process returns a zero value. This way, we only spend $O(k_i)$ time to generate $k_i$. 
Then, at each iteration, $k_i$ random nodes are chosen. To do this in time $O(k_i)$, we can copy the nodes of $T$ into an array, which is done at the beginning only once and hence takes $O(n)$ time in total, and keep a separation index $s$ that separates the selected elements from the ones that have not been selected yet, as explained in the proof of Lemma \ref{growingsubtree-runtime}. Adding a chosen node to $T_i$, $L$ and $B$ can be done in $O(1)$. At the end of $n$ such iterations, a total of $\sum_{i =1 }^{n} k_i$ random choices are made and this is clearly less than or equal to $\sum_{i=1}^{n} |V(T_i)|$. 

Now, it remains to show that generation of $n$ subtrees can also be done in time $O(\sum_{i=1}^{n} |V(T_i)|)$ once $T_i$, $L$ and $B$ are populated with initial randomly selected nodes. Our aim is to construct the minimal subtree of $T$ connecting all the nodes in $\{x_1,\ldots, x_{k_i}\}$. To this end, at every iteration, we add the parents of all nodes of highest level to subtree $T_i$ and replace these nodes by their parents in $L$. The way we store the nodes in $L$ enables us to access unprocessed nodes at a given level in constant time. However, to be able to start with the highest level in constant time initially, we need to know the highest level of the randomly selected $k_i$ nodes, which can be found in $O(k_i)$. While processing some node $x$ at level $d$, we first investigate whether its parent node $P(x)$ has already been included in $L$ by checking index $P(x)$ of $B$ in constant time. If it is one, it means that the parent node has already been included in  $L$ and $T_i$, and we do not do anything. Otherwise, we 
append $P(x)$ to the list at index $(d-1)$ of $L$ and set the corresponding index of $B$ to one. When done with $x$, we remove it from $L$ in $O(1)$ since $x$ is an 
element of a list, and set index $x$ of $B$ to zero, which is again $O(1)$. Recall that since $T_i$ is represented as a list, adding an element to $T_i$ can be done in $O(1)$. Thus, we perform constant-time operations for each node under consideration in the inner for loop. 

At the beginning of the while loop $L$ has $k_i$ isolated nodes. Whenever two nodes in $L$ have a common parent, the cardinality of $L$ decreases by one at the next step. Noting that $|L|$ indicates the number of currently existing connected components, which are to be attached together to reveal the subtree, the while loop to add new nodes to $T_i$ terminates when $|L|=1$; that is, as soon as the minimal subtree has been found. Now, it is enough to notice that during the generation of each subtree $T_i$ using this method, we only consider and add the nodes of $T$ which are in $T_i$, and iterate only through the levels that are contained in $T_i$. In other words, $|L|$ becomes 1 and the while loop stops after exactly when $|V(T_i)|$ nodes are considered. Because we spend constant time for each of the $|V(T_i)|$ nodes, the overall complexity of the operations within the while loop becomes $O(|V(T_i)|)$ for each subtree $T_i$. In order to obtain $O(|V(T_i)|)$ for the entire loop, we need to ensure that termination condition of the while loop can be checked in constant 
time. For this purpose, we keep the number of nodes in $L$ as a variable, incrementing whenever a new node is added and decrementing upon removal of a node, which takes at most $O(|V(T_i)|)$ time. Finally, the arrays $L$ and $B$ should be reset before being passed to the next subtree. We know that $L$ will contain a single node at the end of the while loop, and equivalently a single nonzero element will exist in $B$. When the loop terminates, we will know at which index (level) of $L$ we were finally at. So, we can simply access the final node, set the element in $B$ corresponding to that node to zero, and delete it from $L$, all in constant time. This way, the algorithm from while loop on to the next subtree completes in time $O(|V(T_i)|)$. In total, these operations add to the running time of Algorithm ConnectingNodes a term of $O(\sum_{i=1}^{n} |V(T_i)|)$. \qed

\end{proof}

Lemma \ref{connectingnodes-runtime}, together with Theorem \ref{runningtime}, gives the following:

\begin{corollary}
	Algorithm ChordalGen, using the ConnectingNode method, runs in time $O(\sum_{i=1}^{n} |V(T_i)|)$.
\end{corollary}

\medskip

We have thus presented two different methods for generating subtrees of a given tree, both of which result in an algorithm for generating random chordal graphs in time$O(\sum_{i=1}^{n} |V(T_i)|)$. In the next subsection we present yet another subtree generation method, having running time $O(n^2+\sum_{i=1}^{n} |V(T_i)|)$. We include this algorithm for the sake of completeness and better comparison basis in our tests in the next section. This algorithm can for example be used when one is interested in generating chordal graphs with predominantly large maximal cliques as the density grows.

\subsection{Algorithm PrunedTree}
\label{section:removal_of_edges}

The input to this algorithm consists of a tree $T$ on $n$ nodes, an edge deletion fraction $f$, which is a rational number between 0 and 1, and a selection barrier $s$, which is a real number between 0 and 1. 
To generate subtree $T_i$, we randomly select a fraction $f$ of the edges on the tree and remove them. The number of edges to delete, say $l$, is calculated as $\lfloor (n-1)f \rfloor$,  which will leave $l+1$ subtrees in total. We then determine the sizes of the $l+1$ subtrees and store the distinct values. We pick a random size $k_i$ from the set of largest $100(1-s) \%$ of distinct values, and randomly choose a subtree with size $k_i$. We repeat this $n$ times to generate all the subtrees. One should note that we could simply select one connected component (subtree) at random without any preferential treatment; however, parameter $s$ makes it easier to increase the density of the chordal graph by favoring larger components as the value of $s$ advises. So, parameter $s$ is an additional means to tune the edge density of the chordal graph; as its value increases, the size of the subtree to be selected tends to increase too. Increasing the edge deletion fraction $f$, however, tends to decrease the average size of subtrees emerging from deletion of edges.

\medskip 

{\bf Algorithm PrunedTree} 

\smallskip
{\bf Input:}  A tree $T$ on $n$ nodes, edge deletion fraction $f$, and selection barrier $s$

{\bf Output:}  A set of $n$ non-empty subtrees of $T$

\smallskip
{\bf for} $i=1$ {\bf to} $n$ {\bf do}

\quad Create a copy $T'$ of $T$

\quad Select randomly $\lfloor (n-1)f \rfloor$ edges of $T'$ and delete them from $T'$

\quad Determine connected components of $T'$ and their sizes

\quad Select randomly a subtree size $k_i$ from the highest $ 100(1-s) \% $ subtree sizes.

\quad Select a random component of size $k_i$ and choose it as $T_i$

\smallskip
Output $\{T_1, T_2, \ldots, T_n\}$ \\

\begin{lemma}
	\label{prunedtree-runtime}
	The running time of Algorithm PrunedTree is $O(n^2)$.
\end{lemma}

\begin{proof}		
	Creating a copy of $T$, deleting a subset of its edges, and computing the resulting connected components takes $O(n)$ time by standard BFS. Now, we create an array $A$ of size $n$, where each element in $A$ is a linked list. For each connected component of $T'$ of size $t$, we add this component at the end of the list in $A[t]$. Clearly, initializing $A$, and adding all subtrees to appropriate lists takes $O(n)$ time. We also make an additional array $B$ which simply stores the sizes of all subtrees, in sorted order. $B$ can be created in time $O(n)$, using $A$. We use $B$ to find the highest $ 100(1-s) \% $ subtree sizes, by simply using the corresponding last portion of $B$. Random selection of a subtree of size $k_i$ is simply done by picking a subtree from the list $A[k_i]$ in $O(1)$ time. Thus every subtree requires $O(n)$ time to generate.

 Repeating this $n$ times, the overall complexity of PrunedTree algorithm amounts to $O(n^2)$.	
	\qed
\end{proof}

We now obtain the following result using Theorem \ref{runningtime} together with Lemma \ref{prunedtree-runtime}.

\begin{corollary}
	Algorithm ChordalGen, using the PrunedTree method, runs in time $O(n^2+\sum_{i=1}^{n} |V(T_i)|)$.
\end{corollary}

Let us conclude this section with a remark which applies to all of the three subtree generation methods.
Algorithm ChordalGen does not guarantee the connectedness of its output graphs, as also revealed by our experimental results in Section \ref{chapter:experiments}. If connectedness is of particular importance and must be achieved, a possible modification to our algorithms can guarantee it without increasing the overall time complexity. 
To this end, we randomly choose one vertex from each connected component of the resulting graph. From each subtree corresponding to the set of vertices selected from the components, we pick one arbitrary node, and we form the last subtree with the union of paths on the host tree that connect these nodes. This way, lastly added vertex $v$ is guaranteed to be linked to at least one vertex from each connected component of $G - v$, and so we ensure the connectedness of the output graph. This
process is very similar to a single iteration of ConnectingNodes subroutine, thus by the proof of Lemma \ref{connectingnodes-runtime}, it takes  $O(\sum_{i=1}^{n} |V(T_i)|)$ time and does not affect the overall complexity of Algorithm ChordalGen. 

\section{Experimental Results}\label{chapter:experiments}

In this section, we give extensive test results to show what kind of chordal graphs are generated by Algorithm ChordalGen. In Tables \ref{tab:GrowingSubtree}-\ref{tab:PrunedTree} we give the experimental results of our presented methods GrowingSubtree, ConnectingNodes and PrunedTree, respectively. We show how the selection of the input parameters affects the number of resulting edges $m$ and connected components (``\# conn. comp.s"). We also present the number of maximal cliques (``\# maximal cliques"), and the minimum, maximum, and mean size for the maximal cliques (``min clique size", ``max clique size", ``mean clique size"), along with their standard deviation (``sd of clique sizes"). For each parameter combination, we performed ten independent runs and report the average values across those ten runs.  For each $n$, we tuned the parameter values in order to approximately achieve some selected average edge density values of 0.01, 0.1, 0.5, and 0.8, where edge density is defined as $\frac{m}{n(n-1)/2}$. We 
made all instances that we present here available at http://www.ie.boun.edu.tr/$\sim$taskin/data/chordal/ where we also offer a broad collection of relatively small-sized chordal graphs on 50 to 500 vertices with varying edge densities.  

Algorithm ChordalGen together with GrowingSubtree is able to output connected chordal graphs unless density is too low, as the results in Table \ref{tab:GrowingSubtree} show. In fact, for average edge density of 0.01, as $n$ increases, the average number of connected components converge to one. If we examine the ``min clique size" column, we see that it is usually one for cases where the average number of connected components is greater than one, suggesting that the reason for obtaining disconnected chordal graphs is largely due to a few isolated vertices and that the dominating rest of the graph is comprised of a connected body. The fact that the starting point of the subtrees is selected uniformly at random and we can directly control the maximum size of them leaves little chance for a set of subtrees not to intersect with any other and so lead to a separate connected component, unless the maximum subtree size $k$ is very small. 

% Table generated by Excel2LaTeX from sheet 'growingSubtree_table'
\begin{table}[H]
	\centering
	\caption{Experimental results of Algorithm ChordalGen with GrowingSubtree method}
		\scalebox{0.8}{
\begin{tabular}{ m{1.3cm} m{1.5cm} m{1.5cm} m{2cm} m{1.5cm} m{1.6cm} m{1.6cm} m{1.6cm} m{1.6cm} m{1.6cm} }
	\toprule
	$n$  &  \parbox[t]{2cm}{max\\subtree\\size ($k$)} & \parbox[t]{2cm}{density} & \parbox[t]{2cm}{$m$} & \parbox[t]{2cm}{\#\\ conn.\\comp.s} & \parbox[t]{2cm}{ \# \\maximal\\cliques} & \parbox[t]{2cm}{ min \\clique \\size} & \parbox[t]{2cm}{max\\ clique \\size} & \parbox[t]{2cm}{mean \\clique \\size} & \parbox[t]{2cm}{sd of\\clique \\sizes} \\
		\midrule
		\multirow{4}[2]{*}{1000} 
		& 7   & 0.011 & 5551.4 & 16.7  & 357.1 & 1.0   & 21.6  & 6.1   & 3.4 \\
		& 33  & 0.104 & 51768.5 & 1.0   & 173.0 & 4.8   & 141.5 & 30.7  & 20.4 \\
		& 139  & 0.497 & 248033.5 & 1.0   & 81.3  & 30.6  & 474.3 & 137.9 & 89.2 \\
		& 324 & 0.803 & 400918.7 & 1.0   & 47.5  & 66.8  & 717.4 & 312.2 & 159.5 \\
		\midrule
		\multirow{4}[2]{*}{2500} 
		& 13   & 0.011 & 34605.4 & 2.7   & 678.3 & 1.2   & 54.5  & 11.5  & 6.8 \\
		& 63  & 0.104 & 326287.0 & 1.0   & 300.4 & 9.7   & 349.9 & 61.7  & 45.0 \\
		& 269 & 0.505 & 1577474.1 & 1.0   & 134.3 & 50.2  & 1177.4 & 292.8 & 207.7 \\
		& 635 & 0.806 & 2518595.5 & 1.0   & 83.3  & 132.6 & 1861.8 & 673.1 & 397.4 \\
		\midrule
		\multirow{4}[2]{*}{5000} 
		& 20  & 0.010 & 129763.8 & 1.6   & 1104.6 & 1.6   & 96.5  & 18.1  & 11.4 \\
		& 100  & 0.104 & 1296493.4 & 1.0   & 474.0 & 15.9  & 695.0 & 103.0 & 80.0 \\
		& 450 & 0.498 & 6226843.9 & 1.0   & 205.7 & 74.8  & 2390.8 & 501.6 & 374.2 \\
		& 1097 & 0.804 & 10053952.1 & 1.0   & 124.1 & 202.0 & 3656.7 & 1220.6 & 741.9 \\
		\midrule
		\multirow{4}[2]{*}{10000} 
		& 31  & 0.010 & 502155.4 & 1.0   & 1754.4 & 3.4   & 199.2 & 29.1  & 19.8 \\
		& 169  & 0.107 & 5362219.2 & 1.0   & 709.8 & 22.2  & 1376.0 & 181.6 & 149.6 \\
		& 751 & 0.506 & 25298684.2 & 1.0   & 304.9 & 104.9 & 4687.5 & 894.0 & 711.5 \\
		& 1855 & 0.802 & 40103196.8 & 1.0   & 184.1 & 278.3 & 7445.0 & 2141.8 & 1459.7 \\
		\bottomrule
	\end{tabular}%
	\label{tab:GrowingSubtree}%
}
\end{table}%

Table \ref{tab:ConnectingNodes} reports the outputs of Algorithm ChordalGen using ConnectingNodes for generating subtrees. As the experimental results given in Table \ref{tab:ConnectingNodes} reveal, Algorithm ConnectingNodes should be input very small $\lambda$ values in order to achieve even quite dense graphs. Since even few number of selected nodes may result in large subtrees, which increases the chances of potential intersections with other subtrees and hence the number of edges in the output graph, the number of selected nodes has to be restricted via low values for $\lambda$, which is the main ingredient in setting the cardinality of node selection. Because of this, the selected node set, and hence the subtree, commonly be comprised of a single node, especially in graphs with low density. Therefore, when there are many single-node subtrees, intersections are not very likely. Thus, we observe many isolated vertices in the generated chordal graphs, as implied by the high number of connected components 
and minimum size of one in maximal cliques (see ``min clique size" column) in ConnectingNodes method. 

% Table generated by Excel2LaTeX from sheet 'connNodes_table'
\begin{table}[H]
	\centering
	\caption{Experimental results of Algorithm ChordalGen with ConnectingNodes method}
	\scalebox{0.8}{
		\begin{tabular}{ m{1.3cm} m{1.4cm} m{1.5cm} m{2.2cm} m{1.5cm} m{1.6cm} m{1.6cm} m{1.6cm} m{1.6cm} m{1.6cm} }
		\toprule
		$n$  &  \parbox[t]{2cm}{$\lambda$} & \parbox[t]{2cm}{density} & \parbox[t]{2cm}{$m$} & \parbox[t]{2cm}{\#\\ conn.\\comp.s} & \parbox[t]{2cm}{ \# \\maximal\\cliques} & \parbox[t]{2cm}{ min \\clique \\size} & \parbox[t]{2cm}{max\\ clique \\size} & \parbox[t]{2cm}{mean \\clique \\size} & \parbox[t]{2cm}{sd of\\clique \\sizes} \\
		\midrule
		\multirow{4}[2]{*}{1000} 
		& 0.5   & 0.011 & 5455.4 & 349.0 & 597.0 & 1.0   & 75.8  & 3.0   & 5.5 \\
		& 1.2   & 0.100 & 49805.1 & 121.4 & 495.3 & 1.0   & 266.5 & 8.0   & 23.1 \\
		& 2.7   & 0.507 & 253074.6 & 8.6   & 238.7 & 1.0   & 627.0 & 30.7  & 87.4 \\
		& 4.1   & 0.804 & 401708.6 & 1.8   & 96.3  & 1.6   & 835.4 & 81.5  & 183.9 \\
		\midrule
		\multirow{4}[2]{*}{2500} 
		& 0.6   & 0.010 & 31559.8 & 849.2 & 1475.8 & 1.0   & 194.6 & 3.5   & 10.1 \\
		& 1.2   & 0.101 & 314818.0 & 298.5 & 1215.3 & 1.0   & 657.5 & 9.6   & 40.7 \\
		& 2.7   & 0.503 & 1571946.8 & 27.9  & 594.0 & 1.0   & 1620.6 & 36.0  & 142.1 \\
		& 4.1   & 0.800 & 2498034.2 & 3.1   & 226.5 & 1.2   & 2074.8 & 102.8 & 315.3 \\
		\midrule
		\multirow{4}[2]{*}{5000} 
		& 0.6   & 0.010 & 127700.5 & 1693.5 & 2960.1 & 1.0   & 395.8 & 3.8   & 15.2 \\
		& 1.2   & 0.103 & 1290089.2 & 578.3 & 2409.8 & 1.0   & 1396.8 & 10.6  & 57.9 \\
		& 2.7   & 0.505 & 6308093.4 & 44.4  & 1148.6 & 1.0   & 3217.5 & 41.0  & 211.0 \\
		& 4.1   & 0.805 & 10060406.5 & 4.7   & 435.7 & 1.0   & 4261.4 & 119.4 & 460.6 \\
		\midrule
		\multirow{4}[2]{*}{10000} 
		& 0.6   & 0.010 & 501760.2 & 3365.6 & 5901.9 & 1.0   & 806.2 & 4.0   & 21.8 \\
		& 1.2   & 0.100 & 5022899.1 & 1180.4 & 4794.0 & 1.0   & 2703.4 & 11.7  & 83.8 \\
		& 2.7   & 0.502 & 25114409.8 & 97.3  & 2300.6 & 1.0   & 6355.1 & 44.0  & 291.5 \\
		& 4.1   & 0.803 & 40154270.6 & 9.3   & 852.9 & 1.0   & 8484.9 & 136.2 & 682.4 \\
		\bottomrule
	\end{tabular}%
	\label{tab:ConnectingNodes}%
}
\end{table}%

Table \ref{tab:PrunedTree} presents the experimental results of Algorithm ChordalGen when used with PrunedTree method. The two columns ``edge del. fr. ($f$)" and ``selection barrier ($s$)" correspond to input parameters that PrunedTree takes as input, whose role are explained in Section \ref{section:removal_of_edges}. Here, we observe that for density values of 0.1, 0.5, and 0.8, the output graphs are predominantly connected. As in the previous two methods, minimum size of maximal cliques in case of 0.01 density is one, implying that the main cause of the number of connected components is probably a small group of isolated vertices.  

% Table generated by Excel2LaTeX from sheet 'prunedTree_table'
\begin{table}[H]
	\caption{Experimental results of Algorithm ChordalGen with PrunedTree method}
	\scalebox{0.8}{
	\begin{tabular}{ m{1.3cm} m{1.5cm} m{1.5cm} m{1.5cm} m{2cm} m{1.5cm} m{1.6cm} m{1.6cm} m{1.6cm} m{1.6cm} m{1.6cm} }
		\toprule
		$n$  &  \parbox[t]{2cm}{edge\\ del. fr.\\($f$)}  & \parbox[t]{2cm}{selection\\ barrier\\($s$)} & \parbox[t]{2cm}{density} & \parbox[t]{2cm}{$m$} & \parbox[t]{2cm}{\#\\ conn.\\comp.s} & \parbox[t]{2cm}{ \# \\maximal\\cliques} & \parbox[t]{2cm}{ min \\clique \\size} & \parbox[t]{2cm}{max\\ clique \\size} & \parbox[t]{2cm}{mean \\clique \\size} & \parbox[t]{2cm}{sd of\\clique \\sizes} \\

		\midrule
		\multirow{4}[2]{*}{1000} 
		& 0.950 & 0.35  & 0.011 & 5619.3 & 45.8  & 324.5 & 1.0   & 30.4  & 5.5   & 4.1 \\
		& 0.700 & 0.60  & 0.104 & 51765.9 & 1.0   & 99.9  & 4.2   & 133.6 & 35.9  & 25.2 \\
		& 0.140 & 0.85  & 0.497 & 248172.1 & 1.0   & 50.2  & 193.0 & 337.1 & 278.3 & 35.2 \\
		& 0.100 & 0.93  & 0.806 & 402349.8 & 1.0   & 36.5  & 397.7 & 621.6 & 530.8 & 53.6 \\
		\midrule
		\multirow{4}[2]{*}{2500} 
		& 0.950 & 0.70  & 0.011 & 34013.9 & 28.3  & 542.5 & 1.0   & 72.3  & 9.5   & 8.8 \\
		& 0.700 & 0.70  & 0.101 & 316270.6 & 1.0   & 150.4 & 5.5   & 335.0 & 70.8  & 58.0 \\
		& 0.120 & 0.90  & 0.507 & 1584225.2 & 1.0   & 66.2  & 492.0 & 844.4 & 703.2 & 84.4 \\
		& 0.077 & 0.95  & 0.801 & 2500840.1 & 1.0   & 56.5  & 996.2 & 1530.1 & 1304.7 & 123.0 \\
		\midrule
		\multirow{4}[2]{*}{5000}
		& 0.950 & 0.77  & 0.010 & 130970.2 & 21.7  & 833.8 & 1.0   & 177.7 & 14.0  & 15.2 \\
		& 0.700 & 0.75  & 0.097 & 1216527.9 & 1.0   & 202.6 & 4.8   & 655.2 & 117.0 & 106.5 \\
		& 0.080 & 0.91  & 0.495 & 6182739.6 & 1.0   & 101.0 & 1083.4 & 1571.7 & 1395.7 & 109.0 \\
		& 0.045 & 0.96  & 0.801 & 10004264.5 & 1.0   & 99.9  & 2269.3 & 2955.8 & 2672.9 & 146.1 \\
		\midrule
		\multirow{4}[2]{*}{10000} 
		& 0.900 & 0.50  & 0.010 & 479501.2 & 22.5  & 1394.8 & 1.0   & 286.6 & 20.6  & 24.5 \\
		& 0.700 & 0.81  & 0.102 & 5076707.1 & 1.0   & 260.0 & 7.2   & 1415.8 & 204.5 & 206.0 \\
		& 0.060 & 0.93  & 0.507 & 25357868.2 & 1.0   & 143.5 & 2359.9 & 3176.8 & 2882.4 & 177.7 \\
		& 0.031 & 0.96  & 0.793 & 39653114.8 & 1.0   & 157.7 & 4705.0 & 5709.5 & 5319.0 & 198.9 \\
		\bottomrule
	\end{tabular}%
	\label{tab:PrunedTree}%
}
\end{table}%

We want to compare our results to the results showing the kind of chordal graphs that are generated by Alg 2a \cite{markenzon2008two}. Note, however that, the results given by \cite{markenzon2008two}  only contain graphs on 10000 vertices, with varying number of edges. Most metrics presented in \cite{markenzon2008two} are about the number of edges. When it comes to the maximal cliques, they present only the average maximum clique size over the generated graphs for each edge density. Comparing these to our numbers we see that graphs corresponding to edge densities 0.01, 0.1, 0.5, and 0.8 of Alg 2a have average maximum clique sizes 727, 2847, 6875, and 8760, respectively. As can be seen from Tables \ref{tab:GrowingSubtree}-\ref{tab:PrunedTree}, these numbers are quite higher than the corresponding numbers for the graphs generated by Algorithm ChordalGen. In fact, studying the numbers more carefully, we can conclude that the maximum clique of a graph generated by Alg 2a contains almost all the edges of the 
graph. In the case of density 0.01, such a clique contains more than half of the edges, whereas in the case of higher densities, the largest clique contains more than 80, 94, and 95 percent of the edges, respectively. Thus, there does not seem to be an even distribution of the sizes of maximal cliques of graphs generated by Alg 2a.

As we mentioned in Section \ref{chapter:literature}, we also implemented Alg 1 \cite{andreou2005generating}, but without imposing a limit on the maximum degree of the output graph, because no detail was given about how the method avoids exceeding a given maximum degree in \cite{andreou2005generating}. 
In Table \ref{tab:Alg1} we give results of Alg 1 analogous to Tables \ref{tab:GrowingSubtree}--\ref{tab:PrunedTree} for 1000, 2500, and 5000 vertices. In order to obtain results for Table \ref{tab:Alg1} comparable to those given in Tables \ref{tab:GrowingSubtree}-\ref{tab:PrunedTree}, we aimed to have approximately the same edge density values. For this purpose, when determining the number of new neighbors of a vertex at each step in Alg 1, we multiplied the total number of candidate vertices with a coefficient between 0 and 1, which we call \textit{upper bound coefficient}.
A running time analysis for this algorithm has not been given in \cite{andreou2005generating}. 
With our implementation, this algorithm turned out to be too slow to allow testing graphs on 10000 vertices in a reasonable amount of time. However, already from the obtained numbers, we can reach a conclusion for Alg 1 similar to that on Alg 2a. Observe that the maximum clique sizes obtained for 5000 vertices by Alg 1, are comparable to the maximum clique sizes obtained for 10000 vertices by Algorithm ChordalGen. Hence, like Alg 2a, also Alg 1 seems to generate graphs with few big maximal cliques. 

As can be seen in Table \ref{tab:Alg1}, Alg 1 outputs connected chordal graphs for the selected set of average edge density values and number of vertices. The minimum size of the maximal cliques did not show much variation throughout our experiments and almost always turned out to be two. The consistency in this measure may be an indication of the lack of potential to produce a diverse range of maximal clique sizes.

% Table generated by Excel2LaTeX from sheet 'peoBased_table'
\begin{table}[H]
	\centering
	\caption{Experimental results of our implementation of Alg 1 \cite{andreou2005generating}}
	\scalebox{0.8}{
	\begin{tabular}{ m{1.3cm} m{1.7cm} m{1.5cm} m{2cm} m{1.5cm} m{1.6cm} m{1.6cm} m{1.6cm} m{1.6cm} m{1.6cm} }
		\toprule
		$n$  &  \parbox[t]{2cm}{upper\\ bound\\coef.} & \parbox[t]{2cm}{density} & \parbox[t]{2cm}{$m$} & \parbox[t]{2cm}{\#\\conn.\\comp.s} & \parbox[t]{2cm}{ \# \\maximal \\cliques} & \parbox[t]{2cm}{min \\clique \\size} & \parbox[t]{2cm}{max\\clique\\size} & \parbox[t]{2cm}{mean\\clique\\size} & \parbox[t]{2cm}{sd of\\clique\\sizes} \\
		\midrule
		\multirow{4}[2]{*}{1000} 
		& 0.00130 & 0.011 & 5659.3 & 1.0   & 933.3 & 2.0   & 58.0  & 4.9   & 9.4 \\
		& 0.00300 & 0.100 & 49717.1 & 1.0   & 753.3 & 2.0   & 219.3 & 28.4  & 60.0 \\
		& 0.01100 & 0.506 & 252864.4 & 1.0   & 401.5 & 2.0   & 562.5 & 190.1 & 233.6 \\
		& 0.03500 & 0.805 & 401945.6 & 1.0   & 191.6 & 2.4   & 788.6 & 399.4 & 342.6 \\
		\midrule
		\multirow{4}[2]{*}{2500} 
		& 0.00053 & 0.011 & 33201.8 & 1.0   & 2320.8 & 2.0   & 154.5 & 8.8   & 26.1 \\
		& 0.00120 & 0.100 & 313001.8 & 1.0   & 1882.3 & 2.0   & 548.8 & 69.3  & 159.3 \\
		& 0.00440 & 0.502 & 1568857.4 & 1.0   & 1006.5 & 2.0   & 1400.6 & 462.7 & 593.3 \\
		& 0.01400 & 0.799 & 2495447.1 & 1.0   & 469.8 & 2.0   & 1975.9 & 936.8 & 888.2 \\
		\midrule
		\multirow{4}[2]{*}{5000} 
		& 0.00027 & 0.011 & 133829.0 & 1.0   & 4629.0 & 2.0   & 313.9 & 15.9  & 56.2 \\
		& 0.00062 & 0.107 & 1339169.7 & 1.0   & 3717.9 & 2.0   & 1136.5 & 147.0 & 342.7 \\
		& 0.00220 & 0.494 & 6179872.9 & 1.0   & 2032.5 & 2.0   & 2774.2 & 897.7 & 1180.8 \\
		& 0.00700 & 0.801 & 10011146.3 & 1.0   & 939.0 & 2.0   & 3950.0 & 1901.6 & 1794.5 \\
		\bottomrule
	\end{tabular}%
	\label{tab:Alg1}%
}
\end{table}%

%n=1000
\begin{figure}[H]
	\centering
	\begin{subfigure}[b]{0.99\textwidth}
		\includegraphics[width=\textwidth]{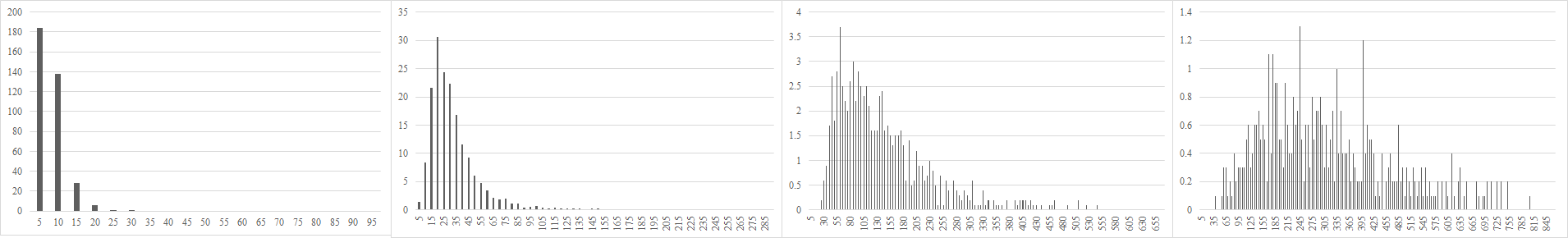}
		\caption{Results from Algorithm ChordalGen with GrowingSubtree method}
		\label{fig:kabakcicegi_1000_4lu}
	\end{subfigure}

	\begin{subfigure}[b]{0.99\textwidth}
		\includegraphics[width=\textwidth]{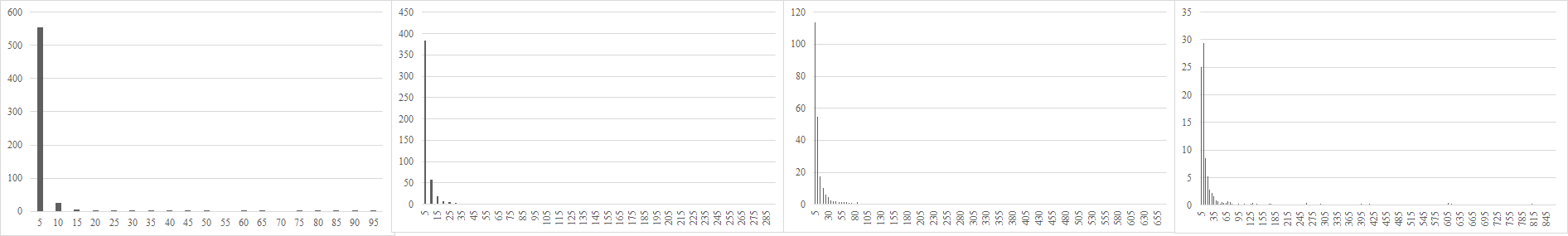}
		\caption{Results from Algorithm ChordalGen with ConnectingNodes method}
		\label{fig:nodeDeletion_1000_4lu}
	\end{subfigure}
	
	\begin{subfigure}[b]{0.99\textwidth}
		\includegraphics[width=\textwidth]{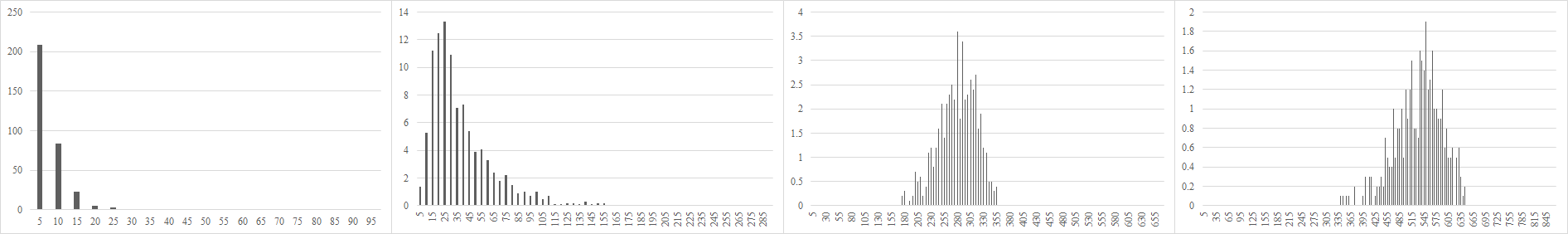}
  		\caption{ Results from Algorithm ChordalGen with PrunedTree method}
		\label{fig:edgeDeletion_1000_4lu}
	\end{subfigure}
	
	\begin{subfigure}[b]{0.99\textwidth}
		\includegraphics[width=\textwidth]{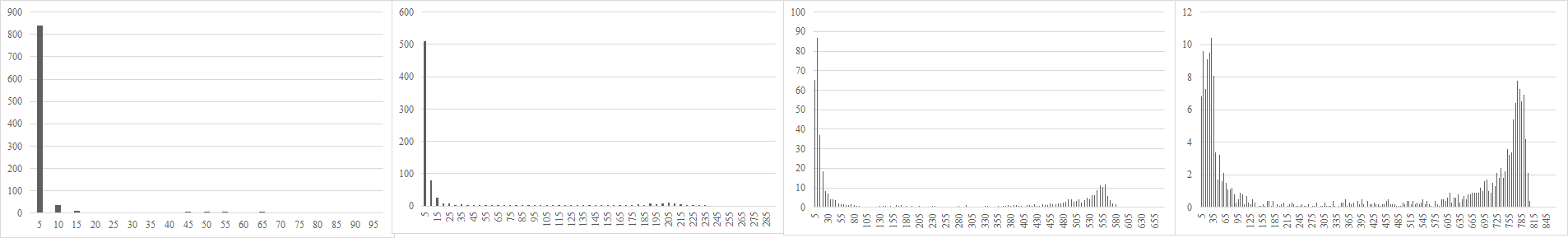}
		\caption{Results from our implementation of Alg 1  \cite{andreou2005generating}}
		\label{fig:peoBased_1000_4lu}
	\end{subfigure}
	\caption{Histograms of maximal clique sizes for $ n = 1000 $ and average edge densities 0.01, 0.1, 0.5, and 0.8 (from left to right)}\label{fig:histograms_1000}
\end{figure}

In our next set of experimental results we investigate how the sizes of the maximal cliques are distributed. 
Figures \ref{fig:histograms_1000}-\ref{fig:histograms_10000} show the average number of maximal cliques across ten independent runs in intervals of width five, for 1000, 2500, 5000, and 10000 vertices and varying edge densities. These figures consist of four subfigures, except Figure \ref{fig:histograms_10000} which contains only the first three, and each subfigure is comprised of four histograms corresponding to four different average edge density values. The first three sub-figures on the top row show the results from Algorithm ChordalGen combined with each one of the three subtree generation methods presented, and the bottom row shows results of our implementation of Alg 1 \cite{andreou2005generating}. For a given $n$ and average density value, the ranges of $ x $-axes are kept the same in order to render the histograms comparable. The $y$-axes, however, have different ranges because maximum frequencies in histograms may vary drastically.

%n=2500
\begin{figure}[H]
	\centering
	\begin{subfigure}[b]{0.99\textwidth}
		\includegraphics[width=\textwidth]{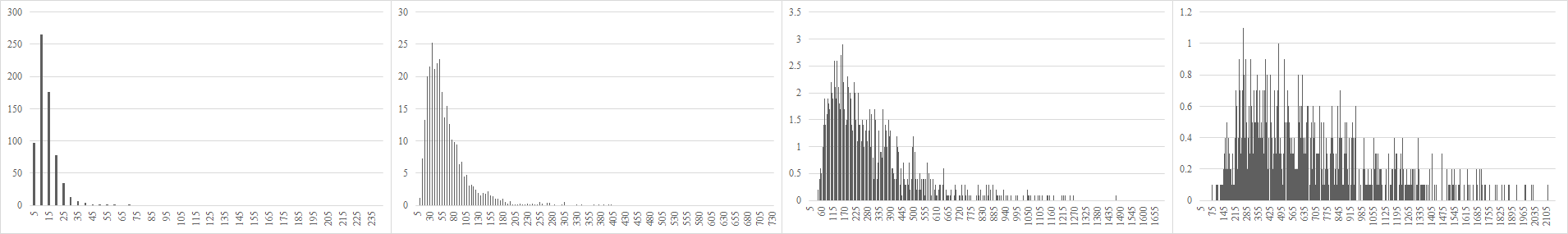}
		\caption{Results from Algorithm ChordalGen with GrowingSubtree method}
		\label{fig:kabakcicegi_2500_4lu}
	\end{subfigure}
	
	\begin{subfigure}[b]{0.99\textwidth}
		\includegraphics[width=\textwidth]{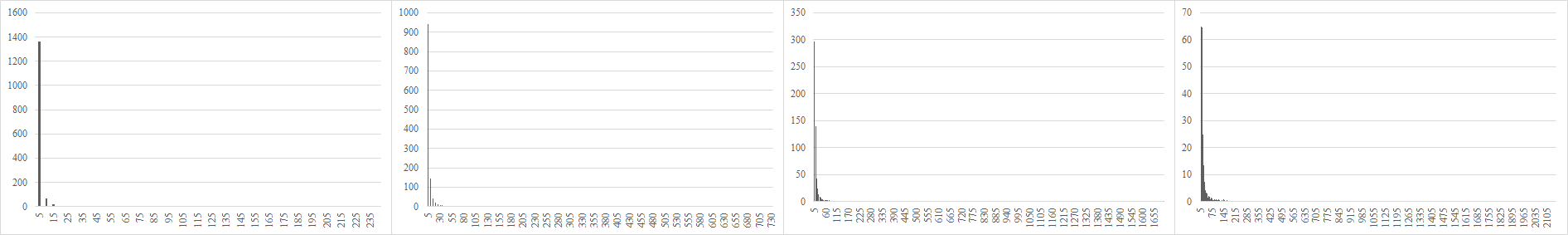}
		\caption{Results from Algorithm ChordalGen with ConnectingNodes method}
		\label{fig:nodeDeletion_2500_4lu}
	\end{subfigure}
	
	\begin{subfigure}[b]{0.99\textwidth}
		\includegraphics[width=\textwidth]{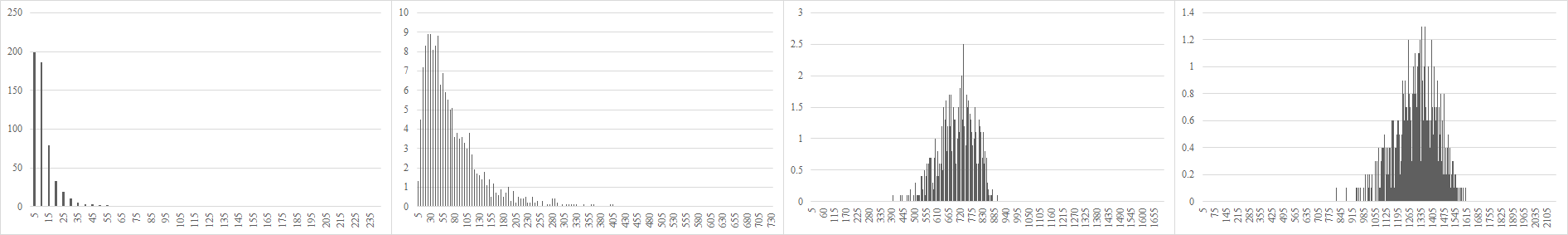}
		\caption{Results from Algorithm ChordalGen with PrunedTree method}
		\label{fig:edgeDeletion_2500_4lu}
	\end{subfigure}
	
	\begin{subfigure}[b]{0.99\textwidth}
		\includegraphics[width=\textwidth]{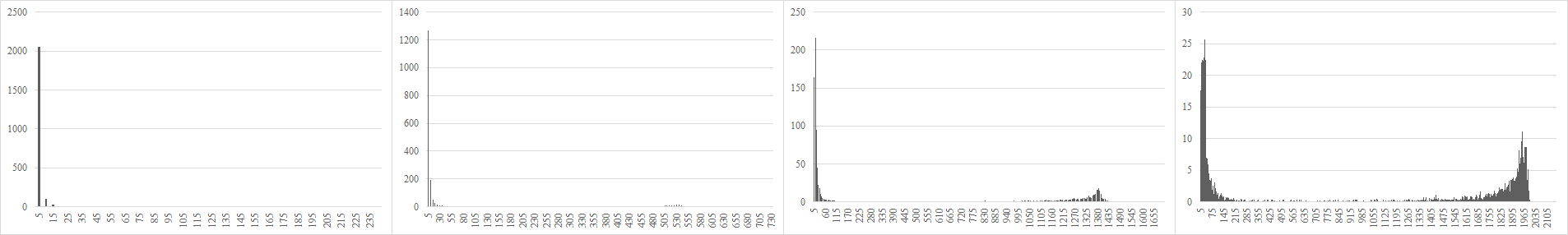}
		\caption{Results from our implementation of Alg 1  \cite{andreou2005generating}}
		\label{fig:peoBased_2500_4lu}
	\end{subfigure}
	\caption{Histograms of maximal clique sizes for $ n = 2500 $ and average edge densities 0.01, 0.1, 0.5, and 0.8 (from left to right)}\label{fig:histograms_2500}
\end{figure}

The sizes of maximal cliques of graphs produced by GrowingSubtree method are dispersed fairly over the range, which becomes more noticeable with the increase in edge densities (as we proceed to the right). Frequencies do not show any obvious bias toward some portion of its domain, which may be considered as an indicator of the diversity of the chordal graphs produced, which is a desired characteristic of a random chordal graph generator. In ConnectingNodes method, however, the vast majority of cliques have size ten or less. The frequencies of larger cliques are barely noticeable compared to the dominant small-sized set. As the graphs become denser, frequencies of relatively larger cliques start to become visible too, but general behaviour remains the same. So, if chordal graphs with predominantly very small clique sizes are sought, ConnectingNodes method can be preferred. In PrunedTree method, the mode of the distribution shifts with the increase in edge densities and the sizes of cliques become clustered 
around some moderate value over the given range. As for Alg 1, the vast majority of maximal cliques of its output graphs have sizes of 2 to 15 for graphs with low densities of 0.01 and 0.1. With the increase in edge densities, frequencies of large-size maximal cliques become visible relative to the dominant small clique frequencies; however, all but the extremes of the range is barely used regardless of selection of $n$ and edge density. 

%n=5000
\begin{figure}[H]
	\centering
	\begin{subfigure}[b]{0.99\textwidth}
		\includegraphics[width=\textwidth]{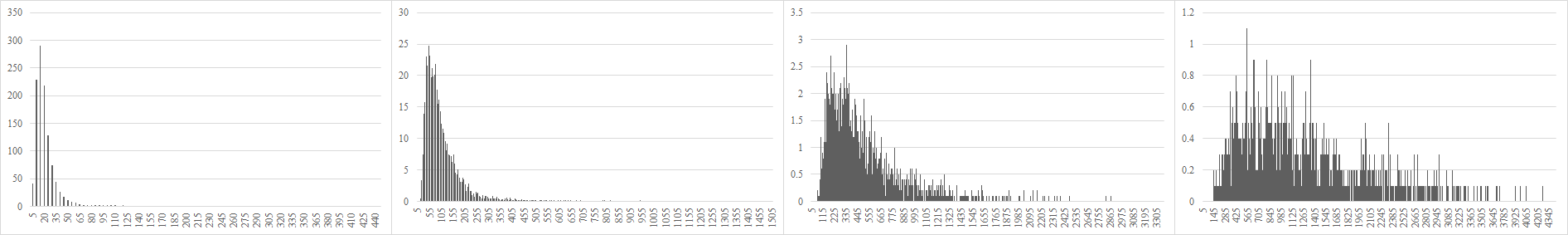}
		\caption{Results from Algorithm ChordalGen with GrowingSubtree method}
		\label{fig:kabakcicegi_5000_4lu}
	\end{subfigure}
	
	\begin{subfigure}[b]{0.99\textwidth}
		\includegraphics[width=\textwidth]{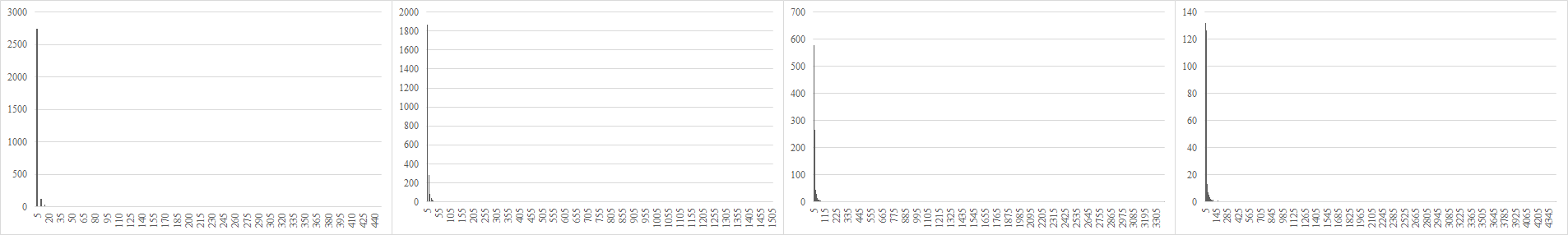}
		\caption{Results from Algorithm ChordalGen with ConnectingNodes method}
		\label{fig:nodeDeletion_5000_4lu}
	\end{subfigure}
	
	\begin{subfigure}[b]{0.99\textwidth}
		\includegraphics[width=\textwidth]{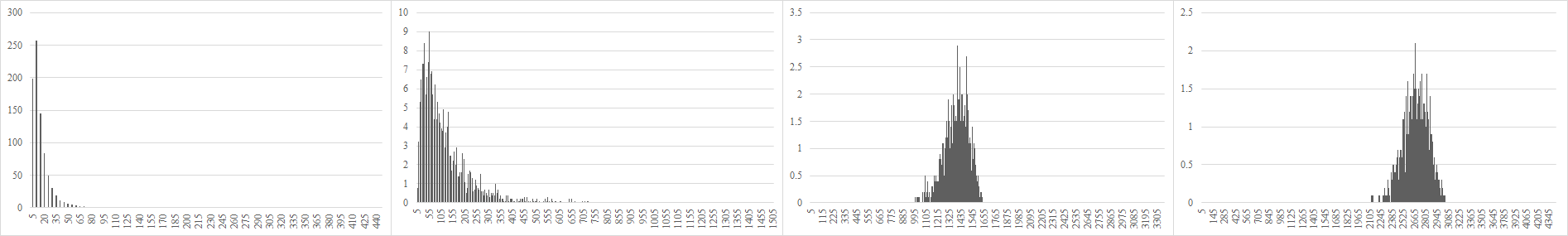}
		\caption{Results from Algorithm ChordalGen with PrunedTree method}
		\label{fig:edgeDeletion_5000_4lu}
	\end{subfigure}
	
	\begin{subfigure}[b]{0.99\textwidth}
		\includegraphics[width=\textwidth]{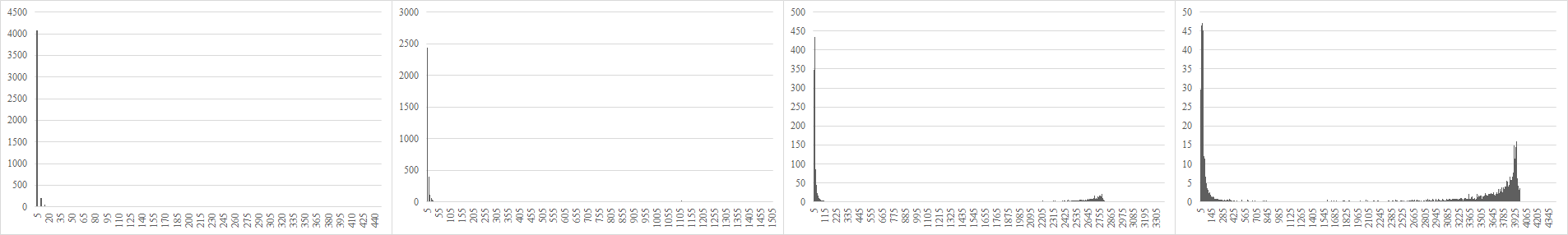}
		\caption{Results from our implementation of Alg 1  \cite{andreou2005generating}}
		\label{fig:peoBased_5000_4lu}
	\end{subfigure}
	\caption{Histograms of maximal clique sizes for $ n = 5000 $ and average edge densities 0.01, 0.1, 0.5, and 0.8 (from left to right)}\label{fig:histograms_5000}
\end{figure}

%n=10000
\begin{figure}[H]
	\centering
	\begin{subfigure}[b]{0.99\textwidth}
		\includegraphics[width=\textwidth]{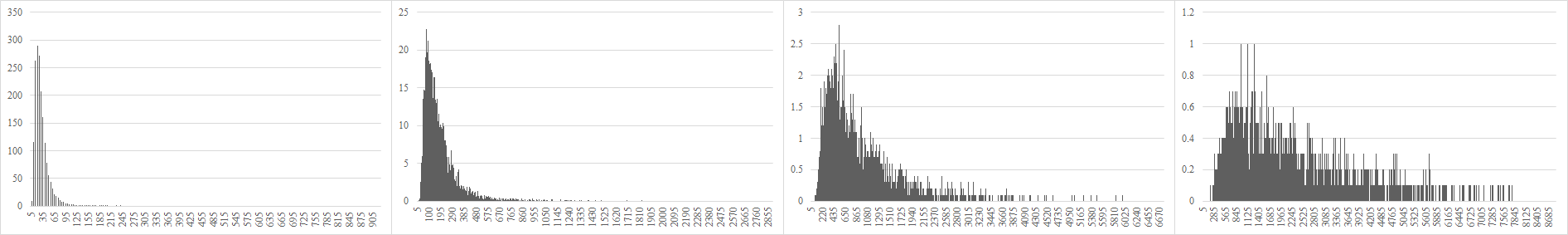}
		\caption{Results from Algorithm ChordalGen with GrowingSubtree method}
		\label{fig:kabakcicegi_10000_4lu}
	\end{subfigure}
	
	\begin{subfigure}[b]{0.99\textwidth}
		\includegraphics[width=\textwidth]{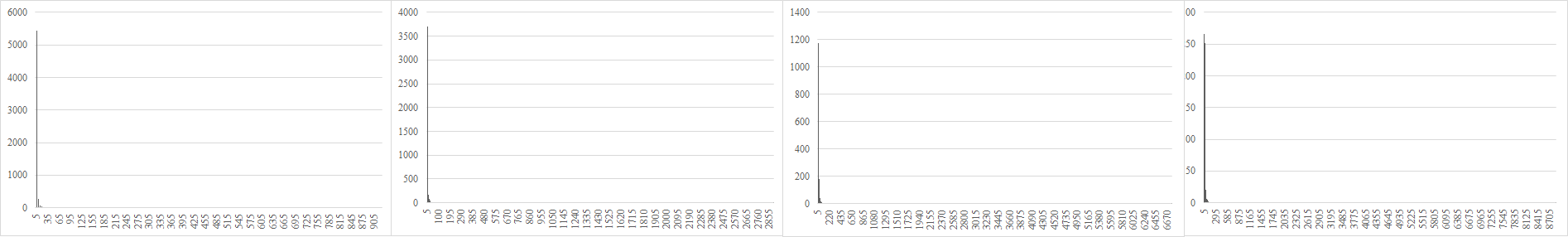}
		\caption{Results from Algorithm ChordalGen with ConnectingNodes method}
		\label{fig:nodeDeletion_10000_4lu}
	\end{subfigure}
	
	\begin{subfigure}[b]{0.99\textwidth}
		\includegraphics[width=\textwidth]{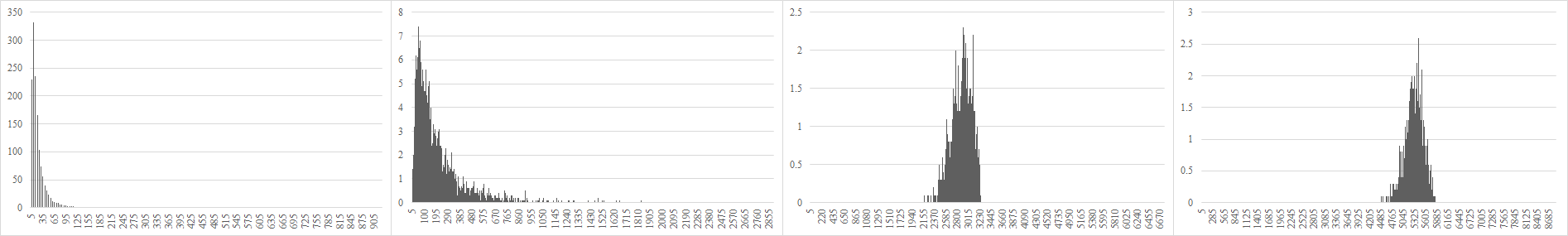}
		\caption{Results from Algorithm ChordalGen with PrunedTree method}
		\label{fig:edgeDeletion_10000_4lu}
	\end{subfigure}

	\caption{Histograms of maximal clique sizes for $ n = 10000 $ and average edge densities 0.01, 0.1, 0.5, and 0.8 (from left to right)}\label{fig:histograms_10000}
\end{figure}

\section{Conclusion}\label{chapter:conc}

To the best of our knowledge, Algorithm ChordalGen is the first algorithm for random chordal graph generation based directly on subtree intersection characterization. It is very general and flexible in the sense that many different methods for subtree generation can be plugged in. 

The three different subtree generation methods presented here each offer output graphs of different structures. As far as the distribution of maximal clique sizes are concerned, ConnectingNodes and PrunedTree methods yield graphs of somewhat more specific structure in the sense that the sizes of maximal cliques are always clustered in the very initial portion in ConnectingNodes method, and in PrunedTree method in the initial part of the range for low densities, in middle portions for moderate to high densities. GrowingSubtree method, though, in addition to its advantageous time complexity, generates the most varied chordal graphs compared to both existing methods and to other two of our suggested methods. Depending on the context and structural needs for the output graph, Algorithm ChordalGen can be used with one of the three subroutines chosen suitably in order to produce chordal graphs of varying size and density. For those who prefer to test various algorithms on chordal graphs without focusing on 
the generation procedure, we offer all instances used in this paper as well as a broad collection of relatively small-sized chordal graphs on 50 to 500 vertices with varying edge densities at http://www.ie.boun.edu.tr/$\sim$taskin/data/chordal/.

Last but not least, our work gives rise to an open question about chordal graphs which has not been addressed to date to the best of our knowledge: what is the worst case time complexity of a chordal graph generation algorithm which produces the entire set of nodes of all subtrees in a subtree intersection model? In the current paper, we give a lower bound, namely $\Omega(mn^{1/4})$, on the time complexity. However, the exact worst case time complexity remains unknown.

\section*{Acknowledgements}

The authors are indebted to Mordechai Shalom for pointing out an error in the running time analysis of Algorithm ChordalGen that appeared in \cite{CIAC2017} by providing the example in the proof of Lemma \ref{lem:key}. This work was initiated while Heggernes was visiting Bo\u{g}azi\c{c}i University, and it is supported by the Research Council of Norway. \c{S}eker, Ekim and Ta\c{s}k{\i}n are supported by the Bo\u{g}azi\c{c}i University Research Fund, Grant 11765, whose support is greatly acknowledged. Ekim is also grateful to Turkish Academy of Science for the GEBIP award which supported this work.


\begin{thebibliography}{10}


\bibitem{Pinar1}
F.~Abu-Khzam, P.~Heggernes.
Enumerating minimal dominating sets in chordal graphs. Inf. Process. Lett. 116(12): 739--743 (2016).


\bibitem{andreou2005generating}
M.~I. Andreou, V.~G. Papadopoulou, P.~G. Spirakis, B.~Theodorides and
A.~Xeros. Generating and radiocoloring families of perfect graphs.
{\em Experimental and Efficient Algorithms\/}, pp. 302--314, Springer, 2005.

%\bibitem{pigeonhole}
%P. E. Black. ``pigeonhole sort", in Dictionary of Algorithms and Data Structures [online], Vreda Pieterse and Paul E. Black, eds. 19 June 2006. (accessed 17.01.2017) Available from: https://www.nist.gov/dads/HTML/pigeonholeSort.html

\bibitem{BP}
J.~R.~S. Blair and B.~W. Peyton.
An Introduction to Chordal Graphs and Clique Trees. 
In {\em Graph Theory and Sparse Matrix Computations},  
IMA Vol. in Math. Appl. 56: 1--27, Springer, 1993.

\bibitem{Sofsem2014}
M.~Bougeret, N.~Bousquet, R.~Giroudeau, and R.~Watrigant.
Parameterized Complexity of the Sparsest k-Subgraph Problem in Chordal Graphs. SOFSEM 2014: 150--161, Springer.

\bibitem{buneman1974characterisation}
P.~Buneman. 
A characterisation of rigid circuit graphs, 
Disc. Math. 9(3): 205--212, 1974.


\bibitem{BLS99}
A.~Brandst\"adt, V.~B. Le, and J.~Spinrad.
{\em Graph Classes: A Survey}.
SIAM Monographs on Discrete Mathematics and Applications (1999).

\bibitem{Dirac}
G.~A. Dirac.
On rigid circuit graphs.
Ann.~Math.~Sem.~Univ.~Hamburg 25: 71--76, 1961.


\bibitem{fulkerson1965incidence}
D.~Fulkerson and O.~Gross. 
Incidence matrices and interval graphs.
Pac. J. of Math. 15(3): 835--855, 1965. 

\bibitem{gavril1972algorithms}
F.~Gavril.
Algorithms for minimum coloring, maximum clique, minimum
covering by cliques, and maximum independent set of a chordal graph.
SIAM J. on Comp. 1(2): 180--187, 1972.


\bibitem{gavril1974intersection}
F.~Gavril.
The intersection graphs of subtrees in trees are exactly
the chordal graphs.
J. of Comb. Th. B, 16(1): 47--56, 1974.

\bibitem{Pinar2}
P.~Golovach, P.~Heggernes, and D.~Kratsch.
Enumerating minimal connected dominating sets in graphs of bounded chordality. Theor. Comput. Sci. 630: 63--75 (2016) 

\bibitem{Pinar}
P.~Golovach, P.~Heggernes, D.~Kratsch, and R.~Saei.
An exact algorithm for Subset Feedback Vertex Set on chordal graphs.
J. of Disc. Alg. 26: 7--15 (2014).   


\bibitem{Gol04}
M.~C. Golumbic.
{\em Algorithmic Graph Theory and Perfect Graphs}.
Annals of Disc.~Math.~57, Elsevier (2004).

\bibitem{Hajnal}
A.~Hajnal and J.~Sur{\'a}nyi.
{\"U}ber die {A}ufl{\"o}sung von {G}raphen in vollst{\"a}ndige
{T}eilgraphen.
{\em Ann. Univ. Sci. Budapest}, pages 113--121, 1958.


\bibitem{heggernes2006minimal}
P.~Heggernes.
Minimal triangulations of graphs: A survey.
Disc. Math. 306(3): 297--317, 2006.

\bibitem{knuth1969art}
D.~E. Knuth. {\em The Art of Computer Programming: Seminumerical
	Algorithms}, volume 2, chapter 4, Addison--Wesley, 1969.

\bibitem{Dagstuhl}
D.~Loksthanov. Dagstuhl Seminar 14071 ``Graph Modification Problems", 2014.

\bibitem{lueker1979linear}
G.~S. Lueker and K.~S. Booth.
A linear time algorithm for deciding interval graph isomorphism.
JACM 26(2): 183--195, 1979.

%\bibitem{lueker1979linear}
%Lueker, G.~S. and K.~S. Booth, \enquote{A linear time algorithm for deciding
%	interval graph isomorphism}, {\em Journal of the ACM (JACM)\/}, Vol.~26,
%No.~2, pp. 183--195, 1979.

\bibitem{markenzon2008two}
L.~Markenzon,  O.~Vernet, and L.~H. Araujo. 
Two methods for the
generation of chordal graphs.
Ann. of Op. Res.
157(1): 47--60, 2008.


\bibitem{DMarx}
D.~Marx.
Parameterized coloring problems on chordal graphs. Theor. Comput. Sci. 351(3): 407--424, 2006.

\bibitem{Neel}
N.~Misra, F.~Panolan, A.~Rai, V.~Raman, S.~Saurabh.
Parameterized Algorithms for Max Colorable Induced Subgraph Problem on Perfect Graphs. WG 2013: 370--381, Springer.

\bibitem{pearl2014probabilistic}
J.~Pearl.
{\em Probabilistic Reasoning in Intelligent Systems: Networks of
	Plausible Inference\/}, Morgan Kaufmann, 2014.

\bibitem{pemmaraju2005approximating}
S.~V. Pemmaraju, S.~Penumatcha and R.~Raman.
Approximating interval
coloring and max-coloring in chordal graphs.
J. of Exp.  Alg. 10: 2--8, 2005.

\bibitem{Rodionov}
A.~S. Rodionov and H.~Choo.
On Generating Random Network Structures: Trees. International Conference on Computational Science 2003: 879--887, LNCS 2658, Springer.

\bibitem{rose1972graph}
D.~J. Rose. A graph-theoretic study of the numerical solution of
sparse positive definite systems of linear equations.
Graph theory and computing 183: 217, 1972.

\bibitem{rose1976algorithmic}
D.~J. Rose, R.~E. Tarjan, and G.~S. Lueker.
Algorithmic aspects of
vertex elimination on graphs.
SIAM J. on Comp. 5(2): 266--283, 1976.

\bibitem{Spinrad03}
J.~P. Spinrad.
{\em Efficient graph representations}.
AMS, Fields Institute Monograph Series 19 (2003).


\bibitem{CIAC2017}
O.~\c{S}eker, P.~Heggernes, T.~Ekim, Z.~C.~Ta\c{s}k{\i}n.
Linear-time generation of random chordal graphs.
Algorithms and Complexity: 10th International Conference, CIAC 2017, LNCS 10236: 442-453, Springer.

%\bibitem{uib}
%A number of master theses at the Department of informatics, University of Bergen, Norway, 2000--2016.

\end{thebibliography}
\end{document}